\documentclass{stacs_proc}

\usepackage{epsfig}
\usepackage{amsmath,latexsym,amssymb
}
\usepackage{algorithm,algorithmic}

\usepackage{pstricks,pst-node,pst-tree}
\usepackage{subfig}
\usepackage{mdwlist}
\usepackage{paralist}

\newcommand{\eqdf}{\stackrel{\rm def}{=}}
\newcommand{\set}[1]{\left\{ #1 \right\}}

\def\Set#1{\left\{#1\right\}}

\newcommand{\paren}[1]{\left( #1 \right)}
\newcommand{\inv}[1]{\frac{1}{#1}}

\newcommand{\argmin}{\ensuremath{\mbox{argmin}}}
\newcommand{\argmax}{\ensuremath{\mbox{argmax}}}
\newcommand{\eps}{\varepsilon}

\newcommand{\Eqr}[1]{Eq.~(\ref{#1})}

\newcommand{\opt}{\textsc{opt}}

\newcommand{\tp}{\tilde{p}}
\newcommand{\oi}{\mathrm{next}}
\newcommand{\ui}{\mathrm{best}}

\newcommand{\cM}{\mathcal{M}}

\newcommand{\E}[1]{\textbf{E} \left[ #1 \right]}

\def\Feasible{\textsc{Feasible}}

\renewcommand{\paragraph}[1]{\par\smallskip\noindent\textbf{#1}}

\begin{document}

\title[Rent, Lease or Buy]{Rent, Lease or Buy: \\
Randomized Algorithms for Multislope Ski Rental}
\author[bgu]{Z.~Lotker}{Zvi Lotker}
\address[bgu]{Dept.\ of Communication Systems Engineering, Ben Gurion
  University, Beer Sheva 84105 , Israel.}

\author[tau]{B.~Patt-Shamir}{Boaz Patt-Shamir}
\address[tau]{School of Electrical Engineering, Tel Aviv University,
  Tel Aviv 69978, Israel. }

\author[haifa]{D.~Rawitz}{Dror Rawitz} 
\address[haifa]{Faculty of
  Science and Science Education \& C.R.I.,
  University of Haifa, Haifa 31905, Israel.}
\email{zvilo@cse.bgu.ac.il, boaz@eng.tau.ac.il, rawitz@cri.haifa.ac.il}

%
%

\thanks{The second author was supported in part by the Israel Science
  Foundation (grant 664/05) and by Israel Ministry of Science and
  Technology Foundation.}


\keywords{competitive analysis; ski rental; randomized algorithms}


\begin{abstract}
  In the Multislope Ski Rental problem, the user needs a certain
  resource for some unknown period of time.  To use the resource, the
  user must subscribe to one of several options, each of which
  consists of a one-time setup cost (``buying price''), and cost
  proportional to the duration of the usage (``rental rate'').  The
  larger the price, the smaller the rent.  The actual usage time is
  determined by an adversary, and the goal of an algorithm is to
  minimize the cost by choosing the best option at any point in time.
  Multislope Ski Rental is a natural generalization of the classical
  Ski Rental problem (where the only options are pure rent and pure
  buy), which is one of the fundamental problems of online computation.
  The Multislope Ski Rental problem is an abstraction of many problems
  where 
  online decisions cannot be modeled by just two options, e.g., power
  management in systems which can be shut down in parts.  In this
  paper we study randomized algorithms for
  Multislope Ski Rental. Our results include 
  the best possible online randomized strategy for any \emph{additive}
  instance, where the cost of switching from one option to another is
  the difference in their buying prices; and an algorithm that
  produces an $e$-competitive randomized strategy for any
  (non-additive) instance.
%
\end{abstract}

\maketitle

\stacsheading{2008}{503-514}{Bordeaux}
\firstpageno{503}

\section{Introduction}

Arguably, the ``rent or buy'' dilemma is the fundamental problem in
online algorithms: intuitively, there is an ongoing game which may end at any
moment, and the question is to commit or not to commit
. Choosing to commit (the `buy' option) implies paying large cost 
immediately, but low overall cost if the game lasts for a long time.
Choosing not to commit (the `rent' option) means high spending rate,
but lower overall cost if the game ends quickly.  This problem was
first abstracted in the ``Ski Rental'' formulation \cite{KMRS88} as
follows. In the buy option, a one-time cost is incurred, and
thereafter usage is free of charge.  In the rent option, the cost is
proportional to usage time, and there is no one-time cost.  The
deterministic solution is straightforward (with competitive factor 2).
In the randomized model, the algorithm chooses a random time to switch
from the rent to the buy option (the adversary is assumed to know the
algorithm but not the actual outcomes of random
experiments).  As is well known, the best possible online strategy for
classical ski rental has
competitive ratio of $\frac{e}{e-1}\approx1.582$.

In many realistic cases, there may be some intermediate options
between the extreme alternatives of pure buy and pure rent: in
general, it may be possible to pay only a part of the buying cost and then pay
only partial rent.  The general problem, 
called here the \emph{Multislope Ski Rental} problem, can be described
as follows. There are several \emph{states} (or \emph{slopes}), where
each state $i$ is characterized by two numbers: a \emph{buying cost}
$b_i$ and a \emph{rental rate} $r_i$ (see Fig.~\ref{Fig:opt}).
Without loss of generality
, we may assume that for all $i$,
$b_i<b_{i+1}$ and $r_i>r_{i+1}$, namely that after ordering the states
in increasing buying costs, the rental rates are decreasing.  The
basic semantics of the multislope problem is natural: to hold the
resource under state $i$ for $t$ time units, the user is charged
$b_i+r_it$ cost units. An adversary gets to choose how long the
game will last, and the task is to minimize total cost until the game
is over. 

The Multislope Ski Rental problem introduces entirely new difficulties
when compared to the classical Ski Rental problem.  Intuitively,
whereas the only question in the classical version is when to buy, in
the multislope version we need also to answer the question of what to
buy.  Another way to see the difficulty is that the number of
potential transitions from one slope to another in a strategy is one
less than the number of slopes, and finding a single point of
transition is qualitatively easier than finding more than one such
point
.

In addition, the possibility of multiple transitions forces us to define
the relation between multiple ``buys.''  Following~\cite{AIS04}, we
distinguish between two natural cases.  In the \emph{additive} case,
buying costs are cumulative, namely to move from state $i$ to state
$j$ we only need to pay the difference in buying prices $b_j-b_i$. In the
\emph{non-additive} case, there is an arbitrarily defined transition cost
$b_{ij}$ for each pair of states $i$ and $j$.  


\paragraph{Our results.}
In this paper we analyze randomized strategies for Multislope Ski
Rental.  (We use the term \emph{strategy} to refer to the procedure
that makes online decisions, and the term \emph{algorithm} to refer to
the procedure that computes strategies.)  Our main focus is the
additive case, and our main result is an efficient algorithm that
computes the best possible randomized online strategy for any given
instance of additive Multislope Ski Rental problem.
%
%
We first give a simpler algorithm which decomposes a $(k+1)$-slope
instance into $k$ two-slope instances, whose
competitive factor is $e\over e-1$.
%
For the non-additive model, we give a  simple 
$e$-competitive randomized strategy. 


\paragraph{Related Work.}
Variants of ski rental are implicit in many online problems.
The classical (two-slope) ski rental problem, where the buying cost of
the first slope and the rental
rate of the second slope are $0$, was introduced in~\cite{KMRS88},
with optimal strategies achieving competitive factors of 2 (deterministic)
 and $\frac{e}{e-1}$ (randomized). Karlin et al.~\cite{KKR03} apply the
randomized strategy to TCP acknowledgment mechanism and other
problems.
The classical ski rental is sometimes called the \emph{leasing}
problem~\cite{BorElY98}.

Azar et al.~\cite{ABFFLR99} consider a problem that
can be viewed as  non-additive multislope ski rental where
slopes become available over time, and obtain 
an online
strategy whose competitive ratio is $4+2\sqrt{2} \approx 6.83$.
Bejerano et al.~\cite{BCN00}, motivated by rerouting in ATM networks,
study the non-additive multislope problem.  They give a deterministic
$4$-competitive strategy, and show that the factor of $4$ holds
assuming only that the slopes are concave, i.e., when the rent in a
slope may decrease with time.
Damaschke~\cite{Dam03} considers a static version of the  problem
from~\cite{ABFFLR99}, namely
non-additive multislope ski rental problem where each slope is bought
``from scratch.''%
\footnote{It can be shown that strategies
that work for this case also work for the general non-additive case
(see Section~\ref{sec:nonadditive}).
}
  For deterministic strategies, \cite{Dam03} gives an upper
bound of $4$ and a lower bound of  ${5+\sqrt{5}\over2}\approx3.618$;
\cite{Dam03}  also  presents a randomized strategy whose
competitive factor is $2/\ln 2 = 2.88$.
As far as we know, Damaschke's strategy is the only randomized
strategy for multislope ski rental to appear in the literature.
%

Irani et al.~\cite{IGS02} present a deterministic $2$-competitive
strategy for the additive model that generalizes the strategy for the
two slopes case.  They motivate their work by energy saving: each
slope corresponds to some partial ``sleep'' mode of the system.
Augustine et al.~\cite{AIS04} present a dynamic program that computes
the best deterministic strategy for non-additive multislope instances.
%
The case where the length of the game is a stochastic variable with
known distribution is also considered in both~\cite{IGS02,AIS04}.

Meyerson \cite{Meyerson05} defines the seemingly related ``parking
permit'' problem, where there are $k$ types of permits of different
costs, such that each permit allows usage for some duration of
time. Meyerson's results indicate that the problems are not very
closely related, at least from the competitive analysis point of view:
It is shown in \cite{Meyerson05} that the competitive ratio of the
parking permit problem is
$\Theta(k)$ and $\Theta(\log k)$ for deterministic and randomized
strategies, respectively. 


\paragraph{Organization.}
The remainder of this paper is organized as follows. In
Section~\ref{sec:model} we define the basic additive model and make a
few preliminary observations.  
In Section~\ref{sec:par} we give a simple algorithm to solve the
multislope problem, and in
Section~\ref{sec:additive} we present our main result: an optimal
online algorithm.  
An $e$-competitive algorithm for the non-additive case is presented in
Section \ref{sec:nonadditive}.
%
%

\begin{figure}
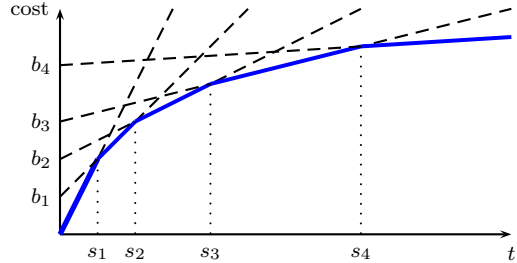

\begin{center}
\begin{scriptsize}
\psset{unit=0.5}
\pspicture(-1,0)(12,6)
\psline[arrows=->](0,0)(12,0)
\rput(12,-0.5){$t$}
\psline[arrows=->](0,0)(0,6)
\rput(-.8,6){cost}
\psline[linestyle=dashed](0,0)(3,6)
\psline[linewidth=2pt,linecolor=blue](0,0)(1,2)
\psline[linestyle=dotted](1,0)(1,2)
\rput(1,-0.5){$s_1$}
\psline[linestyle=dashed](0,1)(5,6)
\psline[linewidth=1.5pt,linecolor=blue](1,2)(2,3)
\rput(-0.5,1){$b_1$}
\psline[linestyle=dotted](2,0)(2,3)
\rput(2,-0.5){$s_2$}
\psline[linestyle=dashed](0,2)(8,6)
\psline[linewidth=1.5pt,linecolor=blue](2,3)(4,4)
\rput(-0.5,2){$b_2$}
\psline[linestyle=dotted](4,0)(4,4)
\rput(4,-0.5){$s_3$}
\psline[linestyle=dashed](0,3)(12,6)
\psline[linewidth=1.5pt,linecolor=blue](4,4)(8,5)
\rput(-0.5,3){$b_3$}
\psline[linestyle=dotted](8,0)(8,5)
\rput(8,-0.5){$s_4$}
\psline[linestyle=dashed](0,4.5)(8,5)
\psline[linewidth=1.5pt,linecolor=blue](8,5)(12,5.25)
\rput(-0.5,4.5){$b_4$}
\endpspicture
\end{scriptsize}
\end{center}
\caption{\em A multislope ski rental instance with 5 slopes: 
The 
thick line indicates the optimal cost as a function of the game
duration time.}
\label{Fig:opt}
\end{figure}

\section{Problem Statement and Preliminary Observations}
\label{sec:model}

In this section we formalize the \emph{additive} version of the
multislope ski rental problem.  A $k$-ski rental instance is defined
by a set of $k+1$ \emph{states}, and for each state $i$ there is a
\emph{buying cost} $b_i$ and a \emph{renting cost} $r_i$.  A state
can be represented by a line: the $i$th state corresponds to the line
$y=b_i+r_ix$. Fig.~\ref{Fig:opt} gives a geometrical interpretation
of a multislope ski rental instance with five states. We 
use the terms ``state'' and 
``slope'' interchangeably.

The requirement of the problem is to specify, for all times $t$, which
slope is chosen at time $t$. We assume that state transitions can be only
forward, and that states cannot be skipped,
i.e., the only transitions allowed are of the type $i\rightarrow i+1$.
We stress that this assumption holds without
loss of generality in the additive model, where a transition from
state $i\rightarrow j$ for $j>i+1$ is equivalent to a sequence of transitions
$i\rightarrow i+1\rightarrow\ldots\rightarrow j$
(cf.~Section~\ref{sec:nonadditive}).  It follows that a
\emph{deterministic strategy} for the additive multislope ski rental
problem is a monotone non-decreasing sequence $(t_1,\ldots,t_k)$ where
$t_i\in[0,\infty)$ corresponds to the transition $i-1\rightarrow i$. 

A \emph{randomized strategy} can be described using a probability
distribution over the family of deterministic strategies.  However, in
this paper we use another way to describe randomized strategies.  We
specify, for all times $t$, a probability distribution over the set of
$k+1$ slopes.  The intuition is that this distribution determines the
actual cost paid by any online strategy.  Formally, a \emph{randomized
profile} (or simply a \emph{profile}) is specified by a vector $p(t) =
(p_0(t),\ldots,p_k(t))$ of $k+1$ functions, where $p_i(t)$ is the
probability to be in state $i$ at time $t$.  The correctness
requirement of a profile is 
$\sum_{i=0}^k p_i(t) = 1$ for all $t \geq 0$.  Clearly, any strategy
is related to some profile.  In the sequel we consider a specific type
of profiles for which a randomized strategy can be easily obtained.

The performance of a profile is defined by its total accrued cost,
which consists of two parts as follows.  Given a randomized profile
$p$, the expected \emph{rental cost} of $p$ at time $t$ is
\[
\textstyle
R_p(t) \eqdf \sum_i p_i(t) \cdot r_i
~~,
\]
and the expected total rental cost up to time $t$ is 
\[
\int_{z=0}^t R_p(z) dz~.
\]
The second part of the cost is the buying cost. In this case it is
easier to define the cumulative buying cost. Specifically, the
expected \emph{total buying
cost} up to time $t$ is
\[
\textstyle
B_p(t) \eqdf \sum_i p_i(t) \cdot b_i~.
\]
The expected \emph{total cost} for $p$ up to time $t$ is
\[
X_p(t) \eqdf B_p(t) + \int_{z=0}^t R_p(z) dz~.
\]
The goal of the algorithm is to minimize total cost up to time $t$ for
any given $t\geq 0$, with respect to the best possible.  Intuitively,
we think of a game that may end at any time.  For any possible ending
time, we compare the total cost of the algorithm with the best
possible (offline) cost.  To this end, consider the optimal solution
of a given instance.  If the games ends at time $t$, the optimal
solution is to select the slope with the least cost at time $t$ (the
thick line in Fig.~\ref{Fig:opt} denotes the optimal cost for any
given $t$).  More formally, the optimal offline cost at time $t$ is
\[
\opt(t) = \min_i (b_i + r_i \cdot t)~.
\]
For $i >0$, denote by $s_i$ the time $t$ instance where $b_{i-1} +
r_{i-1} \cdot t = b_i + r_i \cdot t$, and  define $s_0 = 0$.  It
follows that the optimal slope for a game ending at time $t$ is the
slope $i$ for which  $t \in [s_i,s_{i+1}]$ (if $t=s_i$ for some $i$
then both slopes $i-1$ and $i$ are optimal).

Finally, let us rule out a few trivial cases.  First, note that if
there are two slopes such that $b_i \leq b_j$ and $r_i \leq r_j$ then
the cost incurred by slope $j$ is never less than the cost incurred
slope $i$, and we may therefore just ignore slope $j$ from the
instance.  Consequently, we will assume henceforth, without loss of
generality, that the states are ordered such that $r_{i-1} > r_i$ and
$b_{i-1} < b_i$ for $1 \leq i \leq k$.

Second, using similar reasoning, note that we may
consider only strategies that are monotone over time with respect to
majorization \cite{majorization}, i.e., strategies such
that for any two times $t\le t'$  we have
\begin{equation}
\label{Eqn:back}
\sum_{i=0}^j p_i(t) \ge \sum_{i=0}^j p_i(t')~.
\end{equation}
Intuitively, \Eqr{Eqn:back} means that there is no point is ``rolling
back'' purchases: if at a given time we have a certain composition of
the slopes, then at any later time the composition of slopes may only
improve. Note that \Eqr{Eqn:back} implies that $B_p$ is monotone
increasing and $R_p$ is monotone decreasing, i.e., over time, the
strategy invests non-negative amounts in buying, resulting in
decreased rental rates.


\section{An $e\over e-1$-Competitive Algorithm}
\label{sec:par}

In this section we describe how to solve the multislope problem by
reducing it to the classical two-slope version, resulting in a
randomized strategy whose competitive factor is $e\over e-1$.  This
result serves as a warm-up and it also gives us a concrete upper bound
on the competitiveness of the algorithm presented in
Section~\ref{sec:additive}.

\paragraph{The case of $r_k=0$.}
Suppose we are given an instance $(b,r)$ with $k+1$ slopes, where
$r_k=0$.  We define the following $k$ instances of the classical
two-slopes ski rental problem: in instance $i$ for $i \in
\{1,\ldots,k\}$, we set
\begin{equation}
  \label{eq:parallel}
\mbox{instance $i$:}~~~~~~~~~
b^i_0=0~\mbox{ and }~  r^i_0 = r_{i-1}-r_i~;~~~ b^i_1=b_i - b_{i-1}
~\mbox{ and }~r^i_1= 0~.
\end{equation}

Observe that $b^i_1 = r^i_0 \cdot s_i$, i.e., the two slopes of the
$i$th instance intersect exactly at $s_i$, their intersection point at
the original multislope instance.  Now, let $\opt(t)$ denote the
optimal offline solution to the original multislope instance, and let
$\opt^i(t)$ denote the optimal solution of the $i$th instance at time
$t$, i.e., $\opt^i(t) = \min \{ b^i_1,r^i_0 \cdot t \}$.  We have the
following.

\begin{lemma}
\label{lem:decomposition}
$\opt(t) = \sum_{i=1}^k \opt^i(t)$.
\end{lemma}
\proof
Consider a time $t$ and let $i(t)$ be the optimal multislope state at
time $t$.  Then,
\begin{eqnarray*}
\sum_{i=1}^k \opt^i(t) 
&=& \sum_{i : s_i \leq t} b^i_1 +
  \sum_{i : s_i > t} r^i_0 \cdot t\\
&=& \sum_{i : s_i \leq t} (b_i - b_{i-1}) +\!\! 
  \sum_{i : s_i > t} (r_{i-1} - r_i) \cdot t
~~=~~ b_{i(t)} + r_{i(t)} \cdot t
~~=~~ \opt(t)~.
\end{eqnarray*}

\vspace*{-1mm}
\qed

Given the decomposition (\ref{eq:parallel}), it is easy to obtain a
strategy for any multislope instance by combining strategies for $k$
classical instances. Specifically, what we do is as follows. Let $p^i$
be the $\frac{e}{e-1}$-competitive profile for the $i$th (two slope)
instance (see~\cite{KMRS88}).  We define a profile $\hat{p}$ for the
multislope instance as follows: $\hat{p}_i(t) = p^i_1(t) -
p^{i+1}_1(t)$ for $i \in \{1,\ldots,k-1 \}$, $\hat{p}_0(t) =
p^1_0(t)$, and $\hat{p}_k(t) = p^k_1(t)$.  We first prove that the
profile is well defined.

\begin{lemma}
\label{lem-well-defined}
\begin{inparaenum}[(1)]
\item \label{lem:monotone}
      $p^i_1(t) \leq p^{i-1}_1(t)$ for every $i \in \{1,\ldots,k\}$
      and time $t$.
\item \label{sumto1} 
      $\sum_{i=0}^k \hat{p}_i(t)=1$~.
\end{inparaenum}
\end{lemma}
\begin{proof}
By the algorithm for classical ski rental, we have that the strategy
for the $i$ instance is $p^i_1(t)= (e^{t \cdot r_0^i/b^i_1}-1)/(e-1)$.
Claim~(\ref{lem:monotone}) of the lemma now follows from that fact
that $b^i_1/r^i_0 = s_i > s_{i-1} = b^{i-1}_1/r^{i-1}_0$ for every $i
\in \{1,\ldots,k\}$. Claim~(\ref{sumto1}) follows from the telescopic
sum
\[
\sum_{i=0}^k \hat{p}_i(t)
= p^1_0(t) + \sum_{i=1}^{k-1} (p^i_1(t) - p^{i+1}_1(t)) + p^k_1(t) =
p^1_0(t) + p^1_1(t) = 1
\ .
\]
\end{proof}

Next, we show how to convert the profile
$\hat{p}$ into a strategy. Note that the strategy uses a single random
experiment, since arbitrary dependence between the different $p_i$s
are allowed.

\begin{lemma}
\label{Lem:profile}
Given $\hat{p}$ one can obtain an online strategy whose profile is
$\hat{p}$.
\end{lemma}
\begin{proof}
Define $\hat{P}_i(t) \eqdf \sum_{j \geq i} \hat{p}_j(t)$ and let $U$
be a random variable that is chosen uniformly from $[0,1]$.  The
strategy is as follows: we move from state $i$ to state $i+1$ when $U
= \hat{P}_i(t)$ for every state $i$.  Namely, the $i$th transition
time $t_i$ is the time $t$ such that $U = \hat{P}_i(t)$.
\end{proof}

Thus we obtain the following:

\begin{theorem}
The expected cost of the strategy defined by $\hat{p}$ is at most
$\frac{e}{e-1}$ times the optimal offline cost.
\end{theorem}
\begin{proof}
We first show that by linearity, the expected cost to the combined
strategy is the sum of the costs to the two-slope strategies, i.e.,
that $X_{\hat p}(t) = \sum_{i=1}^k X_{p^i}(t)$. For example, the
buying cost is
\[
B_{\hat p}(t)
= \sum_{i=0}^k \hat{p}_i(t) \cdot b_i
= \sum_{i=0}^{k-1} (p^i_1(t) - p^{i+1}_1(t)) \cdot b_i +
  p^k_1(t) \cdot b_k
= \sum_{i=1}^k p^i_1(t) \cdot (b_i - b_{i-1})
= \sum_{i=1}^k B_{p^i}(t)
\ .
\]
Similarly, $R_{\hat p}(t) = \sum_{i=1}^k R_{p^i}(t)$ by linearity, and
therefore,
\[
X_{\hat p}(t)
= B_{\hat p}(t) + \int_{z=0}^t R_{\hat p}(z) dz
= \sum_{i=1}^k B_{p^i}(t) + 
  \int_{z=0}^t \left( \sum_{i=1}^k R_{p^i}(z) \right) dz
= \sum_{i=1}^k X_{p^i}(t)~.
\]
Finally, by Lemma ~\ref{lem:decomposition} and the fact that the strategies
$p^1,\ldots,p^k$ are $\frac{e}{e-1}$-competitive we conclude  that
\[
X_{\hat p}(t) 
=    \sum_{i=1}^k X_{p^i}(t)
\leq \sum_{i=1}^k \frac{e}{e-1} \cdot \opt^i(t)
=    \frac{e}{e-1} \cdot \opt(t)
\]
which means that $\hat p$ is $\frac{e}{e-1}$-competitive.
\end{proof}

\paragraph{The case of $r_k>0$.}
We note that if the smallest rental rate $r_k$ is positive, then the
competitive ratio is strictly less that $\frac{e}{e-1}$: this can be
seen by considering a new instance where $r_k$ is subtracted from all
rental rates, i.e., $b'_i = b_i$ and $r'_i = r_i - r_k$ for all $0\leq
i \leq k$.  Suppose $p$ is $\frac{e}{e-1}$-competitive
with respect to $(r',k')$ (note that $r'_k=0$). Then the competitive
ratio of $p$ at time $t$ w.r.t.\ the original instance is:
\[
c(t)
=    \frac{X_p(t)}{\opt(t)}
= \frac{X'_p(t) + r_k \cdot t}{\opt'(t) + r_k \cdot t}
\le    \frac{\frac{e}{e-1} \cdot \opt'(t) + r_k \cdot t}{\opt'(t) + r_k
\cdot t}
=    \frac{e}{e-1} -
     \frac{1}{e-1} \cdot \frac{1}{\frac{\opt'(t)}{r_k \cdot t} + 1}
\]
$\frac{d}{dt} \opt'(t) = r_i-r_k$ for $t \in [s_{i-1},s_i)$.  Hence,
the ratio $\frac{\opt'(t)}{r_k \cdot t}$ is monotone decreasing, and
thus $c(t)$ is monotone decreasing as well.
It follows that
\[
c
\leq \frac{e}{e-1} -
     \frac{1}{e-1} \cdot \frac{1}{\frac{r_0-r_k}{r_k} + 1}
=    \frac{e - r_k/r_0}{e-1}
\]
Observe that $c = \frac{e}{e-1}$ when $r_k = 0$, and that $c = 1$ when
$r_k=r_0$ (i.e., when $k=0$). 

\section{An Optimal Online Algorithm}
\label{sec:additive}

In this section we develop an optimal online strategy for any given
additive multislope ski rental instance.  We
reduce the set of all possible strategies to a subset of much simpler
strategies, which on one hand contains an optimal strategy, and on the
other hand is easier to analyze, and in particular, allows us to
effectively find such an optimal strategy.

Consider an arbitrary profile. (Recall that we assume w.l.o.g.\ that
no slope is completely dominated by another.)  As a first
simplification, we confine ourselves to profiles where each $p_i$ has
only finitely many discontinuities.  This allows us to avoid
measure-theoretic pathologies without ruling out any reasonable
solution within the Church-Turing computational model.  It can be
shown that
we may
consider only continuous profiles (details omitted).

So let such a profile $p=(p_0,\ldots,p_k)$ be given.  We show that it
can be transformed into a profile of a certain structure without increasing
the competitive factor.  Our chain of transformations is as follows.
First, we show that it suffices to consider only simple profiles we
call ``prudent.''  Prudent strategies buy slopes in order, one by one,
without skipping and without buying more than one slope at a time.
 We then define the concept of ``tight'' profiles, which
are prudent profiles that spend money at a fixed rate relative to the
optimal offline strategy.  We prove that there exists a tight optimal
profile.  Furthermore, the best tight profile can be effectively
computed: Given a constant $c$, we show how to check whether there
exists a tight $c$-competitive strategy, and this way, using binary
search on $c$, we can find the best tight strategy.  Finally, we
explain how to construct that profile and a corresponding strategy.



\subsection{Prudent and Tight Profiles}
\label{ssec-prudent}

Our main simplification step is to show that it is
sufficient to consider only profiles that buy slopes consecutively one
by one.  Formally, \emph{prudent} profiles are defined as follows.

\begin{definition}[active slopes, prudent profiles]
  A slope $i$ is \emph{active} at time $t$ if $p_i(t)>0$.  A profile
  is called \emph{prudent} if at all times there is either one or two
  consecutive active slopes.
\end{definition}

At any given time $t$, at least one slope is active because
$\sum_ip_i(t)=1$ by the problem definition.  Considering
\Eqr{Eqn:back} as well, we see that a continuous prudent profile
progresses from one slope to next without skipping any slope in
between: once slope $i$ is fully ``paid for'' (i.e., $p_i(t)=1$), the
algorithm will start buying slope $i+1$.

We now prove that the set of continuous prudent profiles contains an
optimal profile.  Intuitively, the idea is that a non-prudent profile must
have two non-consecutive slopes with positive probability at some
time.  In this case we can ``shift'' some probability toward a middle
slope and only improve the overall cost.

\begin{theorem}
\label{thm-prudent}
If there exists a continuous $c$-competitive profile $p$ for some $c
\geq 1$, then there exists a prudent $c$-competitive profile $\tp$.
\end{theorem}
\begin{proof}
Let $p=(p_0,\ldots,p_k)$ be a profile and suppose that all the $p_i$s
are continuous.  It follows that $B_p$ is also continuous.  Define
$\ui(t) = \max \Set{i : b_i \leq B_p(t)}$ and $\oi(t) = \min \Set{i :
b_i \geq B_p(t)}$.  In words, $\ui(t)$ is the most expensive slope
that is fully within the buying budget of $p$ at time $t$, and
$\oi(t)$ is the most expensive slope that is at least partially within
the buying budget of $p$ at time $t$.  Obviously, $\ui(t) \leq \oi(t)
\leq \ui(t)+1$ for all $t$.
Now, we define $\tp$ as follows:
\[
\tp_i(t) = 
\begin{cases}
\displaystyle \frac{b_{\oi}-B_p(t)}{b_{\oi}-b_{\ui}} 
& i = \ui(t) \textrm{ and }\ui(t)\ne\oi(t), \\
 \frac{B_p(t)-b_{\ui}}{b_{\oi}-b_{\ui}} 
& i = \oi(t) \textrm{ and }\ui(t)\ne\oi(t), \\
1 & i = \ui(t) = \oi(t), \\
0 & \textrm{otherwise.}
\end{cases}
\] 
It is not hard to verify that $\sum_i p_i(t)=1$ for every time $t$.
Furthermore, observe that $\tp$ is prudent, because $B_p$ is
continuous.  It remains to show that $\tp$ is $c$-competitive.  We do
so by proving that $B_{\tp}(t) = B_p(t)$ and $R_{\tp}(t) \leq R_p(t)$
for all $t$.  First, directly from definitions we have
\begin{eqnarray*}
B_{\tp}(t) 
&=& p_{\ui(t)}(t) \cdot b_{\ui(t)} + p_{\oi(t)}(t) \cdot b_{\oi(t)} \\
&=& \frac{b_{\oi(t)}-B_p(t)}{b_{\oi(t)}-b_{\ui(t)}} \cdot b_{\ui(t)} +
    \frac{B_p(t)-b_{\ui(t)}}{b_{\oi(t)}-b_{\ui(t)}} \cdot b_{\oi(t)}
~~=~~ B_p(t)~. 
\end{eqnarray*}

Consider now rental payments.  To prove that $R_{\tp}(t) \leq R_p(t)$
for every time $t$ we construct inductively a sequence of probability
distributions $p=p^0,\ldots,p^\ell=\tp$.  The first distribution $p^0$
is defined to be $p$. Suppose now that $p^j$ is not prudent.
Distribution $p^{j+1}$ is obtained from $p^j$ as follows. For any $t$
such that there are two non-consecutive slopes with positive
probability, let $i_1(t),i_2(t),i_3(t)$ be any three slopes such that
$i_1(t) = \argmin \{i : p^j_i(t) > 0\}$, $i_3(t) = \argmax \{i :
p^j_i(t) > 0\}$, and $i_1(t) < i_2(t) < i_3(t)$ (such $i_2(t)$ exists
because $p^j$ is not prudent).  Define
\[
p^{j+1}_i(t) =
\begin{cases}
 p^j_i(t) - \frac{\Delta^j(t)}{b_{i_2(t)}-b_{i_1(t)}} & i = i_1(t), \\
 p^j_i(t) + \frac{\Delta^j(t)}{b_{i_2(t)}-b_{i_1(t)}}
      + \frac{\Delta^j(t)}{b_{i_3(t)}-b_{i_2(t)}}    & i = i_2(t), \\
 p^j_i(t) - \frac{\Delta^j(t)}{b_{i_3(t)}-b_{i_2(t)}} & i = i_3(t), \\
p^j_i(t)                     & i \not\in \set{i_1(t),i_2(t),i_3(t)}
\end{cases}
\]
where $\Delta^j(t) > 0$ is maximized so that $p^{j+1}_i(t) \geq 0$ for
all $i$.  Intuitively, we shift a maximal amount of probability mass
from slopes $i_1(t)$ and $i_3(t)$ to the middle slope $i_2(t)$. The
fact that $\Delta^j(t)$ is maximized means that we have either that
$p^{j+1}_{i_1}(t)=0$, or $p^{j+1}_{i_3}(t)=0$, or both. In any case,
we may already conclude that $\ell < k$.  Also note that by
construction, for all $t$ we have $B_{p^{j+1}}(t) = \sum_i
p^{j+1}_i(t) \cdot b_i = \sum_i p^j_i(t) \cdot b_i = B_{p^j}(t)$.
Hence, $p^\ell=\tp$.

As to the rental cost, fix a time $t$, and consider now  the rent
paid by $p^j$ and $p^{j+1}$:
\begin{eqnarray*}
R_{p^j}(t)\!\!\!\! &-&\!\!\!\! R_{p^{j+1}}(t)~~=~~\\
& =& r_{i_1(t)} \frac{\Delta^j(t)}{b_{i_2(t)}-b_{i_1(t)}}- 
    r_{i_2(t)} \paren{\frac{\Delta^j(t)}{b_{i_2(t)}-b_{i_1(t)}} 
                      \frac{\Delta^j(t)}{b_{i_3(t)}-b_{i_2(t)}}} + 
 r_{i_3(t)} \frac{\Delta^j(t)}{b_{i_3(t)}-b_{i_2(t)}}\\
& =& \Delta^j(t) \cdot 
    \paren{\frac{r_{i_1(t)} - r_{i_2(t)}}{b_{i_2(t)} - b_{i_1(t)}} -
           \frac{r_{i_2(t)} - r_{i_3(t)}}{b_{i_3(t)} - b_{i_2(t)}}} 
~~>~~ 0  
\end{eqnarray*}
where the last inequality follows from the fact that if $i < j$, then
$\frac{b_j - b_i}{r_i - r_j}$ is the $x$ coordinate of the
intersection point between the slopes $i$ and $j$.
\end{proof}


Our next step is to consider profiles that invest in buying as much as
possible under some spending rate cap.  Our approach is motivated by
the following intuitive observation.

\begin{observation}
\label{Obs:tight}
Let $p^1$ and $p^2$ be two randomized prudent profiles.  If
$B_{p^1}(t) \geq B_{p^2}(t)$ for every $t$, then $R_{p^1}(t) \leq
R_{p^2}(t)$ for every $t$.
\end{observation}
In other words, investing available funds in buying as soon as
possible results in lower rent, and therefore in more available funds.
Hence, we define a class of profiles which spend money as soon as
possible in buying, as long as there is a better slope to buy, namely
as long as $p_k(t)<1$.

\begin{definition}
Let $c \geq 1$.  A prudent $c$-competitive profile $p$ is called
\emph{tight} if $X_p(t) = c \cdot \opt(t)$ for all $t$ with $p_k(t) < 1$.
\end{definition}

Clearly, if the last slope is flat, i.e., $r_k=0$, then it must be the
case that $p_k(s_k)=1$ for any profile with finite competitive factor:
otherwise, the cost to the profile will grow without bound while the
optimal cost remains constant.  However, it is important to note that
if $r_k>0$, there may exist an optimal profile $p$ that never buys the
last slope, but still its expected spending rate as $t$ tends to
infinity is $c \cdot r_k$. 

It is easy to see that a tight profile can achieve any achievable
competitive factor.

\begin{lemma}
\label{Lem:tight}
If there exists a $c$-competitive prudent profile $p$ for some $c \geq
1$, then there exists a $c$-competitive tight profile $\tp$.
\end{lemma}
\begin{proof}
Let $\tp$ be the prudent profile satisfying $X_{\tp}(t)=c \cdot
\opt(t)$ for all $t$ for which $\tp_k(t) < 1$.
%
We need to show that $\tp$ is feasible.  Since by definition, $p$ buys
with any amount left, it suffices to show that for all $t$, the rent
paid by $p$ is at most $c \cdot \frac{d}{dt} \opt(t)$.  Indeed,
$R_{\tp}(t) \leq R_p(t)$ for every $t$ due to
Observation~\ref{Obs:tight}, and since $p$ is $c$-competitive it
follows that $R_p(t) \leq c \cdot \frac{d}{dt}\opt(t)$ and we are
done.
\end{proof}



\subsection{Constructing Optimal Online Strategies}
\label{ssec-alg}

We now use the results above to construct an algorithm that produces
the best possible online strategy for the multislope problem.  The
idea is to guess a competitive factor $c$, and then try to construct a
$c$-competitive tight profile.  Given a way to test for success, we
can apply binary search to find the optimal competitive ratio $c$ to
any desired precision.

The main questions are how to test whether a given $c$ is feasible,
and how to construct the profiles.  We answer these questions
together: given $c$, we construct a tight $c$-competitive profile
until either we fail (because $c$ was too small) or until we can
guarantee success.  In the remainder of this section we describe how
to construct a tight profile $p$ for a given competitive factor $c$.

We begin with analyzing the way a tight profile may spend money.
Consider the situation at some time $t$ such that $p_k(t)<1$.  Let $j$
be the maximum index such that $s_j \leq t$.  Then $\frac{d}{dt}
\opt(t) = r_j$.  Therefore, the spending rate of a tight profile at
time $t$ must be $c \cdot r_j$.  If $j<k$, the tight profile may spend
at rate $c \cdot r_j$ until time $s_{j+1}$ (or until $p_k(t)=1$), and
if $j=k$ the tight profile may continue spending at this rate forever.
Hence, for $t \in (s_j,s_{j+1})$, we have
\begin{equation}
\label{eq-tight}
\frac{d}{dt}B_p(t) + R_p(t) 
= c \cdot \frac{d}{dt}\opt(t) 
= c \cdot r_j~.
\end{equation}
Since $p$ is tight and therefore prudent, we also have, assuming
$\ui(t)=i$ and $\oi(t)=i+1$, that
$$B_p(t) = p_{i}(t) b_{i} + p_{i+1}(t) b_{i+1}~, $$
{and}
$$
R_p(t) = p_{i}(t) r_{i} + p_{i+1}(t) r_{i+1}~.
$$
Plugging the above equations into \Eqr{eq-tight}, we get
\[
\frac{d}{dt}p_i(t) b_i + 
\frac{d}{dt}p_{i+1}(t) b_{i+1} + 
p_i(t) r_i + 
p_{i+1}(t) r_{i+1} 
= c \cdot r_j
\]
Since $p$ is prudent, $p_i(t) = 1- p_{i+1}(t)$ and hence
$\frac{d}{dt}p_i(t) = - \frac{d}{dt} p_{i+1}(t)$.  It follows that
\begin{align}
\label{Eqn:differ}
\frac{d}{dt}p_{i+1}(t) + p_{i+1}(t) \cdot \frac{r_{i+1} - r_i}{b_{i+1}-b_i}
= \frac{c \cdot r_j - r_i}{b_{i+1}-b_i}
\end{align}
A solution to a differential equation of the form $y'(x) + \alpha y(x)
= \beta$ where $\alpha$ and $\beta$ are constants is $y =
\frac{\beta}{\alpha} + \Gamma \cdot e^{-\alpha x}$, where $\Gamma$
depends on the boundary condition.  Hence in our case we conclude that
\begin{equation}
  \label{eq:piece}
  p_{i+1}(t) 
  = \frac{c \cdot r_j - r_i}{r_{i+1}-r_i} + 
  \Gamma \cdot e^{\frac{r_i-r_{i+1}}{b_{i+1}-b_i} \cdot t} ~, 
\end{equation}
and $p_i(t) = 1 - p_{i+1}(t)$, where the constant $\Gamma$ is
determined by the boundary condition.

\Eqr{eq:piece} is our tool to construct $p$ in a piecewise 
iterative fashion.  For example, we start constructing $p$ from $t=0$
using
$p_1(t) 
= \frac{c \cdot r_0 - r_0}{r_1-r_0} + 
  \Gamma\cdot e^{\frac{r_0-r_1}{b_1-b_0} \cdot t}$
and the boundary condition $p_1(0)=0$. We get that $\Gamma = 
\frac{r_0(c-1)}{r_0-r_1}$, i.e., 
\[
p_1(t) 
= \frac{r_0(c-1)}{r_0-r_1} \cdot (e^{\frac{r_0-r_1}{b_1-b_0} t} - 1)~,
\]
and this holds for all $t\le\min(s_1,t_1)$, where $t_1$ is the
solution to  $p_1(t_1)=1$.

In general, \Eqr{Eqn:differ} remains true so long as there is no
change in the spending rate and in the slope the profile $p$ is
buying.  The spending rate changes when $t$ crosses $s_j$, and the
profile starts buying slope $i+2$ when $p_{i+1}(t)=1$.

\begin{algorithm}[t]
\begin{small}
\caption{-- $\Feasible(c,\cM)$: true if the $k$-ski instance
  $\cM=(b,r)$ admits competitive factor $c$}
\label{alg-feasible}
\begin{algorithmic}[1]
  \STATE Let $s_i = \frac{b_i-b_{i-1}}{r_{i-1}-r_i}$ for each
         $1 \leq i \leq k$
  \STATE Boundary\_Condition $\gets$ ``$p_1(0)=0$''
  \STATE $j \gets 0$; $i \gets 1$
  \LOOP
    \STATE Define $p_{i}(t) = 
           \frac{c \cdot r_j - r_{i-1}}{r_{i}-r_{i-1}} + 
           \Gamma \cdot 
           \exp({\frac{r_{i-1}-r_{i}}{b_{i}-b_{i-1}} \cdot t})$
    \STATE Try to solve for $\Gamma$ using Boundary\_Condition
    \STATE \textbf{if} no solution \textbf{then return} \ \textsc{false}
        \COMMENT{possible escape if not feasible}
    \STATE $y\gets p_{i}(s_{j})$
    \IF{$y<1$}
      \STATE Boundary\_Condition $\gets$ ``$p_{i}(s_{j})=y$''
      \STATE $j\gets j+1$
        \COMMENT{continue at the next interval $[s_j,s_{j_1}]$}
    \ELSE
      \STATE Let $x$ be such that $p_{i}(x)=1$
      \STATE Boundary\_Condition $\gets$ ``$p_{i+1}(x)=0$''
      \STATE $i\gets i+1$
        \COMMENT{move to next slope}
    \ENDIF
    \STATE \textbf{if} $i>k$ or $j\ge k$ \textbf{then return} \ \textsc{true} 
        \COMMENT{we're done}
  \ENDLOOP
\end{algorithmic}
\end{small}
\end{algorithm}

We can now describe our algorithm.  Given a ratio $c$,
Algorithm~\Feasible\ is able to construct the tight profile $p$ or to
determine that such a profile does not exist.  It starts with the
boundary condition $p_1(0)=0$ and reveals the first part of the
profile as shown above.  Then, each time the spending rate changes or
there is a change in $\ui(i)$ it moves to the next differential
equation with a new boundary condition.  After at most $2k$ such
iterations it either computes a $c$-competitive tight profile $p$ or
discovers that such a profile is infeasible.  Since we are able to
test for success using Algorithm~\Feasible, we can apply binary search
to find the optimal competitive ratio to any desired precision. 

We note that it is easy to construct a strategy that corresponds to
any given prudent profile $p$, as described in the proof of
Lemma~\ref{Lem:profile}.
We conclude with the following theorem.

\begin{theorem}
  There exists an $O(k \log \inv{\eps})$ time algorithm that given an
  instance of the additive multislope ski rental problem for which the
  optimal randomized strategy has competitive ratio $c$, computes 
  a $(c+\eps)$-competitive strategy.
\end{theorem}

\section{An $e$-Competitive Strategy for the Non-Additive Case}
\label{sec:nonadditive}

In this section we consider the non-additive multislope ski rental
problem.  We present a simple randomized strategy which improves the
best known competitive ratio from $2/\ln 2 = 2.88$ to $e$. Our
technique is a simple application of  randomized repeated doubling
(see, e.g.,~\cite{Gal80}), used 
extensively in competitive analysis of online algorithms. For example,
deterministic repeated doubling appears in~\cite{AAFPW97}, and a 
randomized version appears in~\cite{MPT94}.


Before presenting the strategy let us consider the details of the non-additive model.
Augustine at el.~\cite{AIS04} define a general
non-additive model in which a transition cost $b_{ij}$ is associated
with every two states $i$ and $j$, and show that one may
assume w.l.o.g.\ that $b_{ij} = 0$ if $i>j$ and that $b_{ij} \leq b_j$
for every $i<j$.  Observe that we may further assume that $b_{ij} =
b_j$ for every $i$ and $j$, since the optimal (offline) strategy
remains the unchanged.  It follows that the strategies
from~\cite{ABFFLR99,BCN00,Dam03} that were designed for the case of
buying slopes ``from scratch'' also work for the general non-additive
case.

We propose using the following iterative online strategy, which is
similar to the one in \cite{Dam03}, except for the choice of the
``doubling factor.'' Specifically, the $j$th
iteration is associated with a bound $B_j$ on $\opt(\tau)$, where
$\tau$ denotes the termination time of the game.  We define $B_1 \eqdf
\opt(s_1) / \alpha^X$, where $\alpha>1$ and $X$ is a chosen at random
uniformly in $[0,1)$.  We also define $B_{j+1} = \alpha \cdot B_j$.
Let $\tau_j = \opt^{-1}(B_j)$ and let $i_j$ be the optimal offline
state at time $\tau_j$.  In case there are two such states, i.e.,
$\tau_j = s_i$ for some $i$, we define $i_j = i-1$.  It follows that
$i_1 = 0$.  In the beginning of the $j$th iteration the online
strategy buys $i_j$ and stays in $i_j$ until the this iteration ends.
The $j$th iteration ends at time $\tau_j$.
Observe that the first iteration starts with $B_1 = \opt(s_1)$, namely
we use slope 0 until $s_1$.
\begin{theorem}
  The expected cost of the strategy described above is at most $e$
  times the optimum.
\end{theorem}
\begin{proof}
Observe that the first iteration starts with $B_1 = \opt(s_1)$, namely
we use slope 0 until $s_1$, and hence, 
if the game ends during the first iteration, i.e., before $s_1 /
\alpha^X$, then the online strategy is optimal. Consider now the case
where the game ends at time
$\tau \geq s_1 / \alpha^X$, and suppose that $\tau \in
[\tau_\ell,\tau_{\ell+1})$ for $\ell>1$.  In this case, the expected cost
of the online strategy is bounded by
\begin{align*}
\E{\sum_{j=1}^\ell \opt(\tau_j) + \opt(\tau)}
& \leq \E{\sum_{j=1}^{\ell+1} \opt(\tau_j)} 
  \leq \E{\frac{\alpha}{\alpha-1} \cdot \opt(\tau_{\ell+1})} \\
&  =    \E{\frac{\alpha^{2-X}}{\alpha-1} \cdot \opt(\tau)} \\
& =    \frac{\alpha}{\alpha-1} \cdot \int_{x=0}^1 \alpha^x dx \cdot
       \opt(\tau) 
  =    \frac{\alpha}{\ln \alpha} \cdot \opt(\tau) 
\end{align*}
By choosing $\alpha=e$ the competitive ratio is $\frac{\alpha}{\ln
\alpha} = e$ as required.
\end{proof}

\vskip-0.3cm
\section*{Acknowledgment}
We thank Seffy Naor and Niv Buchbinder for stimulating discussions.


\end{document}